\documentclass[11pt]{amsart}

\usepackage[margin=1in]{geometry}
\usepackage{amsmath,amssymb,amsthm,mathtools}
\usepackage{booktabs}
\usepackage[numbers,sort&compress]{natbib}
\usepackage[colorlinks=true,linkcolor=blue,citecolor=blue,urlcolor=blue,
  pdftitle={A Linear Bound on the Rainbow Cycle Number and Approximate EFX},
  pdfauthor={Varun Sivashankar}]{hyperref}

\newtheorem{theorem}{Theorem}

\newtheorem{proposition}[theorem]{Proposition}
\newtheorem{lemma}[theorem]{Lemma}
\theoremstyle{remark}

\newcommand{\efx}{\ensuremath{\mathrm{EFX}}}
\newcommand{\Z}{\mathbb Z}

\title{A Linear Bound on the Rainbow Cycle Number and Approximate EFX}
\author{Varun Sivashankar}
\address{Department of Mathematics, Princeton University,
Princeton, New Jersey, USA}
\email{varunsiva@princeton.edu}
\date{}

\begin{document}

\begin{abstract}
It is open whether every fair-division instance with additive valuations
admits a complete envy-free-up-to-any-good (EFX) allocation.  A well-studied
relaxation allows some goods to remain unallocated and asks for
$(1-\varepsilon)$-EFX.  The rainbow cycle number $R(d)$ was introduced to
study this problem: upper bounds on $R(d)$ yield approximate EFX allocations
with few unallocated goods.  The best previous bound,
$R(d)=O(d\log d)$, gives $O_\varepsilon(\sqrt{n\log n})$ unallocated goods.
We resolve the conjecture that $R(d)$ is linear by proving $R(d)<ed$.  It
follows that every instance with $n$ agents admits a partial
$(1-\varepsilon)$-EFX allocation with $O(\sqrt{n/\varepsilon})$ unallocated
goods.  This is the best possible asymptotic guarantee on the number of
unallocated goods obtainable from the rainbow-cycle reduction.  We also give
a randomized algorithm that finds such an allocation in expected time
polynomial in the input size and $1/\varepsilon$.
\end{abstract}

\maketitle

\section{Introduction}

Consider $n$ agents with nonnegative additive valuations over a set of
indivisible goods.  A partial allocation $(X_1,\ldots,X_n)$ is
$\alpha$-\efx{} if, for every pair of agents $i,j$ and every good $g\in X_j$,
\[
  v_i(X_i)\ge \alpha\,v_i(X_j\setminus\{g\}).
\]
The case $\alpha=1$ is envy-freeness up to any good, or \efx.  Whether a
complete \efx{} allocation always exists for additive valuations is a central
open problem in discrete fair division.  The notion was introduced under this
name by \citet{CaragiannisKKMPSW}.  Complete \efx{} allocations are known to
exist for two agents with general monotone valuations \citep{PlautRoughgarden}.
For three agents, \citet{ChaudhuryGM} proved existence for additive valuations,
and the problem remains open for four or more additive agents; see
\citet{AmanatidisSurvey} for a survey.

One relaxation keeps all goods but weakens the fairness guarantee: a complete
$(\phi-1)$-\efx{} allocation, where $\phi$ is the golden ratio, can be found
in polynomial time for additive valuations \citep{AmanatidisMN}.  Another
allows goods to remain unallocated; this is often called \emph{\efx{} with
charity} \citep{CaragiannisGH}.  \citet{ChaudhuryKMS} proved that exact
\efx{} is possible with at most $n-1$ unallocated goods.
\citet{BergerCFF} improved this to $n-2$ for additive valuations, and
\citet{Mahara} extended the $n-2$ guarantee to general monotone valuations.
Such partial allocations can also retain substantial Nash welfare
\citep{CaragiannisGH}.

A line of work combines these two relaxations: it asks for
$(1-\varepsilon)$-\efx{} while allowing some goods to remain unallocated,
with the aim of leaving far fewer goods unallocated than is currently known
for exact \efx.  \citet{ChaudhuryGMMM} introduced an interesting
combinatorial problem called the rainbow cycle number, denoted by $R(d)$, and
showed how upper bounds on $R(d)$ give bounds on the number of unallocated
goods.  Successive improvements to $R(d)$ produced successively better
\efx{} guarantees
\citep{BerendsohnBK,AkramiACGMM,ChashmJahanSSS}.  Before this work, the best
bound, $R(d)=O(d\log d)$, gave
$O_\varepsilon(\sqrt{n\log n})$ unallocated goods.  We prove the
asymptotically optimal linear bound $R(d)=\Theta(d)$, with the explicit upper
bound $R(d)<ed$, resolving the conjecture that $R(d)=O(d)$.  This gives the
best possible asymptotic guarantee on the number of unallocated goods using
the reduction of \citet{ChaudhuryGMMM}.

\begin{samepage}
\begin{theorem}\label{thm:efx-main}
For every $\varepsilon\in(0,1/2]$, every additive instance with $n$ agents
admits a partial $(1-\varepsilon)$\nobreakdash-\efx{} allocation with
{\small
\[
  O\!\left(\sqrt{\frac{n}{\varepsilon}}\right)
\]
}
unallocated goods.  A randomized algorithm finds one in expected time
polynomial in the input size and $1/\varepsilon$.
\end{theorem}
\end{samepage}

We now define the rainbow cycle number.  For a positive integer $d$, let
$R(d)$ be the largest $k$ for which there is a directed $k$-partite graph
with nonempty vertex classes $V_1,\ldots,V_k$, each of size at most $d$, such
that every vertex has an in-neighbor in every other class and there is no
directed cycle using at most one vertex from each class.  Such a cycle is
called \emph{rainbow}.  In symbols, the in-neighbor condition is
\begin{equation}\label{eq:in-neighbor}
  N_G^-(v)\cap V_i\ne\varnothing
  \qquad\text{for every $i\ne j$ and every $v\in V_j$}.
\end{equation}
Equivalently, every class may be assumed to have size exactly $d$: pad a
smaller class with vertices of outdegree zero, each given an in-neighbor in
every other class.  This preserves \eqref{eq:in-neighbor} and creates
no directed cycle.

In their original paper, \citet{ChaudhuryGMMM} proved
$d\le R(d)\le d^4+d$ and conjectured that $R(d)=O(d)$;
\citet{BerendsohnBK} obtained $R(d)\le d^{2+o(1)}$ and later stated the
stronger conjecture $R(d)=d$.  Two independent works
then proved $R(d)=O(d\log d)$: \citet{AkramiACGMM} used a probabilistic
argument, while \citet{ChashmJahanSSS} introduced the rainbow path degree and
used a lower bound on it to obtain the same upper bound on $R(d)$.
\citet{KirchwegerSzeider} used SAT modulo symmetries to extend the known
equality $R(3)=3$ to $R(4)=4$.
Table~\ref{tab:progress} shows the corresponding progress for approximate
\efx{} when $\varepsilon$ is fixed.

A related problem requires the in-neighbor maps between classes to be
permutations.  Its rainbow cycle number satisfies $R_p(d)\le2d-4$ for $d\ge4$
\citep{ChashmJahanSSS}.  The \efx{} reduction produces arbitrary multipartite
graphs with the in-neighbor property, so this bound does not apply.

\begin{table}[t]
  \centering
  {
  \small
  \begin{tabular}{@{}p{0.62\textwidth}p{0.28\textwidth}@{}}
    \toprule
    Bound on $R(d)$ & Unallocated goods \\
    \midrule
    $d^4+d$ \citep{ChaudhuryGMMM} & $O(n^{4/5})$ \\
    $d^{2+o(1)}$ \citep{BerendsohnBK} & $n^{2/3+o(1)}$ \\
    {\raggedright $O(d\log d)$
      \citep{AkramiACGMM,ChashmJahanSSS}\par}
      & $O(\sqrt{n\log n})$ \\
    $\boldsymbol{ed}$ (this paper) & $\boldsymbol{O(\sqrt n)}$ \\
    \bottomrule
  \end{tabular}
  }
  \caption{Bounds from the rainbow-cycle reduction.}
  \label{tab:progress}
\end{table}

\subsection{Our Contributions}

\begin{itemize}
\item For every directed $k$-partite graph with classes
$V_1,\ldots,V_k$ satisfying the in-neighbor condition and containing no
rainbow cycle, we prove the packing inequality
\[
  (k-1)!\sum_{i=1}^k |V_i|\le\prod_{i=1}^k|V_i|.
\]
If every class has size at most $d$, it implies $k!\le d^{k-1}$ and hence
$R(d)<ed$, resolving the conjecture that $R(d)=O(d)$.  Together with the
known lower bound $R(d)\ge d$, this determines $R(d)$ up to a constant factor
(Theorem~\ref{thm:main}).

\item Combined with the reduction of \citet{ChaudhuryGMMM}, our bound gives a
partial $(1-\varepsilon)$-\efx{} allocation with
$O(\sqrt{n/\varepsilon})$ unallocated goods.  For fixed $\varepsilon$, this
removes the $\sqrt{\log n}$ factor from the previous guarantee.

\item We give a randomized algorithm that finds a rainbow cycle in total
expected $O(k^2)$ time after preprocessing whenever $k>ed$.  This yields an
expected polynomial-time algorithm for finding the approximate \efx{}
allocation above (Lemma~\ref{lem:randomized-recovery}).

\item The same counting idea proves $H(\ell)=\Theta(\ell^2)$ for the rainbow
path degree introduced by \citet{ChashmJahanSSS}
(Theorem~\ref{thm:path-degree}).

\item We discuss consequences of our bound for other \efx{} variants whose
guarantees depend on $R(d)$.  We also show that it gives a short alternate
proof of the known linear bound for zero-sum cycles in finite groups.
\end{itemize}

\section{From Rainbow Cycles to Approximate EFX}

We use the following result of \citet{ChaudhuryGMMM}.

\begin{proposition}\label{prop:transfer}
For every positive integer $d$ and every $\varepsilon\in(0,1/2]$, an additive
instance with $n$ agents admits a partial
$(1-\varepsilon)$\nobreakdash-\efx{}
allocation
whose set $P$ of unallocated goods satisfies
\[
  |P|\le\frac{2n}{\varepsilon d}+R(d).
\]
Moreover, let $T(d)$ be an integer upper bound on $R(d)$.  If a rainbow cycle
can be found in expected polynomial time in every directed graph with more
than $T(d)$ classes of size at most $d$ satisfying
\eqref{eq:in-neighbor}, then an allocation satisfying
\[
  |P|\le\frac{2n}{\varepsilon d}+T(d)
\]
can be found in expected polynomial time.
\end{proposition}

The first statement follows from the proof of their Theorem~3.  The second
follows from their Lemmas~7 and~8 and the termination argument in
Section~4.2.  Their procedure makes only polynomially many updates, and each
update takes polynomial time once the required rainbow cycle has been found.
It follows that the full procedure has expected polynomial running time
whenever the rainbow-cycle subroutine does.

We deduce Theorem~\ref{thm:efx-main} here from Proposition~\ref{prop:transfer},
assuming the two rainbow-cycle results proved in the next section.

\begin{proof}[Proof of Theorem~\ref{thm:efx-main}]
Set
\[
  d=\left\lceil\sqrt{n/\varepsilon}\right\rceil
  \qquad\text{and}\qquad
  T(d)=\lceil ed\rceil.
\]
Theorem~\ref{thm:main} gives $R(d)<ed\le T(d)$, while
Lemma~\ref{lem:randomized-recovery} finds a rainbow cycle in expected
polynomial time whenever the number of classes exceeds $T(d)$.  Applying
Proposition~\ref{prop:transfer} therefore gives
\[
  |P|\le\frac{2n}{\varepsilon d}+T(d)
  =O\!\left(\sqrt{\frac n\varepsilon}\right).
\]
\end{proof}

\paragraph{How the reduction works.}
The procedure of \citet{ChaudhuryGMMM} may start from any partial
$(1-\varepsilon)$-\efx{} allocation $X$.  It preserves
$(1-\varepsilon)$-\efx{} and never decreases any agent's value.  For
existence, $X$ may simply be the empty allocation.

Here is the main idea behind the bound in
Proposition~\ref{prop:transfer}.  At the current allocation $X$, call an
unallocated good $g$ valuable to agent $i$ if
\[
  v_i(\{g\})>\varepsilon v_i(X_i).
\]
Fix $d$ and divide the set $P$ of unallocated goods into high-demand goods
$H$, which are valuable to more than $d$ agents, and low-demand goods $L$,
which are valuable to at most $d$ agents.  When none of the update rules of
\citet{ChaudhuryGMMM} applies, every agent finds at most $2/\varepsilon$
unallocated goods valuable.  Double counting the pairs $(i,g)$ for which
$g$ is valuable to $i$ therefore gives
\[
  |H|\le \frac{2n}{\varepsilon d}.
\]
For the low-demand goods, one can define a multipartite directed graph with
one part for each good in $L$ and at most $d$ vertices in each part (called
the group champion graph by \citet{ChaudhuryGMMM}).  This graph satisfies the
in-neighbor condition in the definition of $R(d)$.  A rainbow cycle would give another
value-nondecreasing update, so the procedure cannot stop while
$|L|>R(d)$.  Consequently,
\[
  |P|=|H|+|L|
  \le \frac{2n}{\varepsilon d}+R(d).
\]

\paragraph{The square-root barrier.}
This formula explains both the improvement in
Theorem~\ref{thm:efx-main} and the limitation of the reduction.  The
previous bound $R(d)=O(d\log d)$ gave
$O_\varepsilon(\sqrt{n\log n})$ unallocated goods, whereas
$R(d)<ed$ gives $O(\sqrt{n/\varepsilon})$.  On the other hand, the known
lower bound $R(d)\ge d$ shows that even the conjectured equality
$R(d)=d$ would leave the expression
\[
  \frac{2n}{\varepsilon d}+d.
\]
Its minimum occurs when $d$ is of order $\sqrt{n/\varepsilon}$.
Thus improving the rainbow cycle number alone cannot make this reduction
give $o(\sqrt{n/\varepsilon})$ unallocated goods; a better bound would
require a different fair-division argument.

\section{The Packing Inequality}

\begin{theorem}\label{thm:main}
Let $G$ be a directed $k$-partite graph with nonempty classes
$V_1,\ldots,V_k$.  Suppose that, for every $i\ne j$ and every $v\in V_j$,
there is an edge $u\to v$ with $u\in V_i$.  If $G$ has no rainbow cycle and
$n_i=|V_i|$, then
\begin{equation}\label{eq:packing}
  (k-1)!\sum_{i=1}^k n_i\le \prod_{i=1}^k n_i.
\end{equation}
\par\noindent
If $n_i\le d$ for every $i$, then
\[
  k!\le d^{k-1}.
\]
In particular, $R(d)<ed$.
\end{theorem}

\begin{proof}
For every $i\ne j$ and every $v\in V_j$, choose one edge $u\to v$ with
$u\in V_i$.  Retaining only these chosen edges cannot create a rainbow cycle,
so we may assume that every vertex has a unique chosen in-neighbor in every
other class.

Fix a permutation $\pi=(\pi_1,\ldots,\pi_k)$ of the classes and a terminal
vertex $v_k\in V_{\pi_k}$.  Working backwards, let
$v_t\in V_{\pi_t}$ be the chosen in-neighbor of $v_{t+1}$ for
$t=k-1,\ldots,1$.  This gives a directed path meeting every class exactly
once:
\[
  v_1\longrightarrow v_2\longrightarrow\cdots\longrightarrow v_k.
\]
Let $T(\pi,v_k):=\{v_1,\ldots,v_k\}$ be its vertex set.

We claim that the map $(\pi,v_k)\mapsto T(\pi,v_k)$ is injective.  Suppose
$T(\pi,v_k)=T(\sigma,w_k)=T$.  If $\pi\ne\sigma$, then some two vertices
$u,v\in T$ occur in opposite orders.  The path defined by one permutation
contains a directed subpath from $u$ to $v$, while the other contains a
directed subpath from $v$ to $u$.  Concatenating these subpaths gives a closed
directed walk, which contains a directed cycle.  Every vertex of this cycle
lies in $T$, so the cycle is rainbow, a contradiction.  Hence $\pi=\sigma$.
Both paths now end in the same class.  Since $T$ contains only one vertex from
that class, $v_k=w_k$.

For a fixed permutation $\pi$, there are $n_{\pi_k}$ choices for the terminal
vertex.  Each class occurs last in exactly $(k-1)!$ permutations, so
injectivity gives
\[
  (k-1)!\sum_{i=1}^k n_i
  \le \prod_{i=1}^k n_i,
\]
which proves \eqref{eq:packing}.

Put $Q=\prod_i n_i$.  The AM--GM inequality gives
\[
  \sum_{i=1}^k n_i\ge kQ^{1/k}.
\]
Combining this with \eqref{eq:packing}, and using
$n_i\le d$, we obtain
\[
  k!\le Q^{(k-1)/k}\le d^{k-1}.
\]
If $k\ge ed$, since $k!>(k/e)^k$, we obtain
\[
  k!>\left(\frac{k}{e}\right)^k\ge d^k\ge d^{k-1},
\]
a contradiction.  Thus $k<ed$, and hence $R(d)<ed$.
\end{proof}

The same proof gives an efficient randomized algorithm for finding the cycle.

\begin{lemma}[Randomized recovery]\label{lem:randomized-recovery}
Let $d\ge2$, and let $G$ satisfy the hypotheses of
Theorem~\ref{thm:main}, with $n_i\le d$ for every $i$.  If $k>ed$, then a
rainbow cycle can be found by a randomized algorithm in expected $O(k^2)$
time, after choosing and storing one in-neighbor from each other class for
every vertex.
\end{lemma}

\begin{proof}
Choose a terminal vertex uniformly from $V(G)$, place its class last, and
order the other classes uniformly at random.  Construct the path
$v_1\to\cdots\to v_k$ as in the proof of Theorem~\ref{thm:main}.  Check if a
directed edge $v_b\to v_a$ exists for any $b>a$.  If so, then
\[
  v_a\longrightarrow\cdots\longrightarrow v_b\longrightarrow v_a
\]
is a rainbow cycle.  If no such edge exists, call the trial unsuccessful.

Unsuccessful pairs consisting of an ordering and a terminal vertex give
distinct transversals.  Indeed, if two gave the same transversal and their
class orders differed, the two paths would contain oppositely directed
subpaths between some pair of vertices.  Their union would contain a directed
cycle, and this cycle would have a backward edge in either path order,
contradicting failure.  If the class orders were the same, the terminal
vertices would also be the same.  Hence, with
$Q=\prod_i n_i$ and $S=\sum_i n_i$,
\[
  \Pr(\text{failure})
  \le \frac{Q}{(k-1)!S}
  \le \frac{d^{k-1}}{k!}
  <\frac1d.
\]
The middle inequality follows from the AM--GM inequality,
as in Theorem~\ref{thm:main}, and the last from
$k!>(k/e)^k>d^k$.  Thus each trial succeeds with probability greater than
$1-1/d$.  The number $N$ of independent trials until success is geometric, so
\[
  \mathbb{E}N<\frac{1}{1-1/d}=\frac{d}{d-1}\le2.
\]
Constructing the path and checking all possible backward edges takes
$O(k^2)$ time per trial.  The total expected running time is therefore
$O(k^2)$.
\end{proof}

We do not know how to derandomize Lemma~\ref{lem:randomized-recovery}; in
particular, the argument does not give a deterministic polynomial-time
algorithm at the linear threshold $k>ed$.

\section{The Rainbow Path Degree}

\citet{ChashmJahanSSS} introduced the rainbow path degree as a tool for
bounding the rainbow cycle number.  They showed that lower bounds on
$H(\ell)$ give upper bounds on $R(d)$, and their estimate
$H(\ell)=\Omega(\ell^2/\log\ell)$ yielded $R(d)=O(d\log d)$.  The same
counting argument, with the terminal vertex fixed, proves that $H(\ell)$ is
quadratic.

\begin{theorem}\label{thm:path-degree}
Let $G$ be an $\ell$-partite directed graph with no rainbow cycle, and suppose
that every vertex has an in-neighbor in every other class.  A directed path is
rainbow if it uses at most one vertex from each class.  Let $f_G(v)$ be the
number of vertices other than $v$ from which there is a directed rainbow path
to $v$.  Let $H(\ell)$ be the minimum of $f_G(v)$ over all such graphs $G$ and
all $v\in V(G)$.
Then, for every $\ell\ge2$,
\[
  H(\ell)\ge
  (\ell-1)\bigl((\ell-1)!\bigr)^{1/(\ell-1)}.
\]
In particular, $H(\ell)=\Theta(\ell^2)$.
\end{theorem}

\begin{proof}
Fix an $\ell$-partite graph $G$ with no rainbow cycle and a vertex
$v\in V_r$.  As above, choose one in-neighbor of every vertex in each other
class.  For $i\ne r$, let
$C_i\subseteq V_i$ be the vertices from which there is a directed rainbow
path to $v$, and put $c_i=|C_i|$.
For each permutation $\pi=(\pi_1,\ldots,\pi_{\ell-1})$ of the classes other
than $V_r$, put $v_\ell=v$ and, working backwards, let
$v_t\in V_{\pi_t}$ be the chosen in-neighbor of $v_{t+1}$.  This produces
the rainbow path
\[
  v_1\longrightarrow\cdots\longrightarrow v_{\ell-1}\longrightarrow v.
\]
Let $T(\pi):=\{v_1,\ldots,v_{\ell-1}\}$; this is a transversal of the sets
$C_i$.

The map $\pi\mapsto T(\pi)$ is injective.  Indeed, if two distinct
permutations produced the same transversal, some two vertices of that
transversal would occur in opposite orders on the two resulting paths.  The
corresponding directed subpaths, one in each direction, would form a closed
directed walk inside the transversal.  This walk contains a directed cycle,
and that cycle is rainbow, a contradiction.  Therefore
\[
  (\ell-1)!\le\prod_{i\ne r}c_i.
\]
The AM--GM inequality now gives
\[
  \begin{aligned}
  f_G(v)
  &=\sum_{i\ne r}c_i\\
  &\ge(\ell-1)\left(\prod_{i\ne r}c_i\right)^{1/(\ell-1)}\\
  &\ge(\ell-1)\bigl((\ell-1)!\bigr)^{1/(\ell-1)}.
  \end{aligned}
\]
Taking the minimum over $G$ and $v$ proves the stated lower bound.  Together
with the bound $H(\ell)\le(\ell-1)(\ell-2)+1$ of
\citet{ChashmJahanSSS}, it gives $H(\ell)=\Theta(\ell^2)$.
\end{proof}

\section{Tightness}

The construction of \citet{ChaudhuryGMMM} gives the lower bound $R(d)\ge d$.
Take $d$ classes, each labeled by $\Z_d$, and include the edges
\[
  \begin{aligned}
  (i,x)&\longrightarrow(j,x) &&\text{if }i<j,\\
  (i,x)&\longrightarrow(j,x+1) &&\text{if }i>j.
  \end{aligned}
\]
Every $(j,y)$ has the in-neighbor $(i,y)$ when $i<j$ and $(i,y-1)$ when
$i>j$.  Suppose there were a rainbow cycle
\[
  (i_1,x_1)\longrightarrow\cdots\longrightarrow(i_r,x_r)
  \longrightarrow(i_1,x_1).
\]
Reading the indices cyclically, the label increases by one at the step
$i_t\to i_{t+1}$ exactly when $i_t>i_{t+1}$.  Thus the total label change is
the number of descents in the cyclic ordering $i_1,\ldots,i_r$.  This number
lies between $1$ and $r-1$, and hence between $1$ and $d-1$.  It is nonzero in
$\Z_d$, contradicting the return to the initial label.  Therefore
\[
  d\le R(d)<ed.
\]
The conjectured value remains $R(d)=d$.

\section{Other EFX Applications}

The rainbow cycle number can also be applied to other variants of \efx.

\begin{itemize}
\item \textbf{Few valuation types.}
Suppose the agents have at most $q$ distinct valuation functions.
\citet{PrakashMehtaNimbhorkar} prove that a partial
$(1-\varepsilon)$-\efx{} allocation exists with
\[
  |P|\le \frac{2q}{\varepsilon d}+R(d).
\]
They use $R(d)=O(d\log d)$ to obtain
$\widetilde O(\sqrt{q/\varepsilon})$ unallocated goods.  Using
$R(d)<ed$ instead removes the logarithmic factor and gives a partial
$(1-\varepsilon)$-\efx{} allocation with
$O(\sqrt{q/\varepsilon})$ unallocated goods.

\item \textbf{Nash welfare.}
For $X=(X_1,\ldots,X_n)$, let
$\operatorname{NW}(X)=(\prod_i v_i(X_i))^{1/n}$.
\citet{FeldmanMaurasPonitka} proved that, for every $\alpha\in[0,1]$, there
is a partial $\alpha$-\efx{} allocation whose Nash welfare is at least
$1/(\alpha+1)$ of the maximum.  Taking $\alpha=1-\varepsilon$ and then
applying the rainbow-cycle reduction gives a partial
$(1-\varepsilon)$-\efx{} allocation with
$O(\sqrt{n/\varepsilon})$ unallocated goods and
\[
  \operatorname{NW}(X)
  \ge \frac{1}{2-\varepsilon}\operatorname{NW}(X^\star).
\]
Here $X^\star$ is a maximum-Nash-welfare allocation.  Starting the
rainbow-cycle procedure from their allocation preserves this inequality,
since no agent's value decreases.
The Nash-welfare factor is due to their theorem; our bound supplies the
guarantee on the number of unallocated goods.
\end{itemize}

\section{Application to Zero-Sum Cycles}

The zero-sum cycle problem was introduced for cyclic groups by
\citet{AlonKrivelevich} and extended to finite abelian groups by
\citet{MeszarosSteiner}.  \citet{BerendsohnBK} observed that it is a special
case of the permutation rainbow cycle problem and treated arbitrary finite
groups.  We recall the reduction here.

Let $\Gamma$ be a finite group, written multiplicatively.  Define
$n(\Gamma)$ to be the least $n$ such that every labeling of the arcs of the
complete bidirected graph on $n$ vertices by elements of $\Gamma$ contains a
directed cycle whose labels, multiplied in order around the cycle, have
product equal to the identity.

This problem is a special case of the rainbow cycle problem.  Given a
$\Gamma$-labeling of the complete bidirected graph on vertices
$1,\ldots,k$, make one class
\[
  V_i=\{(i,x):x\in\Gamma\}
\]
for each vertex $i$.  If the arc $i\to j$ has label $a_{ij}$, include the
edges
\[
  (i,x)\longrightarrow(j,xa_{ij})
  \qquad (x\in\Gamma).
\]
Right multiplication by $a_{ij}$ is a permutation of $\Gamma$, so every
vertex of $V_j$ has exactly one in-neighbor in $V_i$.  Moreover, a rainbow
cycle
\[
  (i_1,x_1)\longrightarrow\cdots\longrightarrow(i_r,x_r)
  \longrightarrow(i_1,x_1)
\]
satisfies
\[
  x_1=x_1a_{i_1i_2}\cdots a_{i_ri_1}.
\]
Cancelling $x_1$ shows that the corresponding directed cycle in the original
graph has label product equal to the identity.  Conversely, every such
zero-sum cycle lifts to a rainbow cycle.  Theorem~\ref{thm:main} therefore
gives
\[
  n(\Gamma)\le R(|\Gamma|)+1<e|\Gamma|+1.
\]

The permutation structure gives a better constant.  Since the maps in the
construction above are permutations, one may instead use the permutation
rainbow number $R_p(d)$, defined by restricting every map between parts to be
a permutation.  The bound $R_p(d)\le2d-2$, proved by both
\citet{BerendsohnBK} and \citet{AkramiACGMM}, gives
\[
  n(\Gamma)\le2|\Gamma|-1
\]
for every nontrivial finite group.  \citet{ChashmJahanSSS} improved this to
$R_p(d)\le2d-4$ when $d\ge4$.  Thus, while our theorem does not improve the
best known constant, it gives a short alternate proof that
$n(\Gamma)=O(|\Gamma|)$ for every finite group.

Earlier, \citet{AlonKrivelevich} proved an $O(q\log q)$ bound for cyclic
groups, and \citet{MeszarosSteiner} proved $n(A)\le8|A|$ for finite abelian
groups.  More recently, \citet{CampbellGHS} announced the sharp bound
$n(\Gamma)\le|\Gamma|+1$ for every finite group and proved it for groups of
odd order; they state that the more involved even-order case will appear in a
second paper.  Meanwhile, \citet{Diwan} obtained a near-sharp estimate for
cyclic groups of arbitrary order.  Sharper bounds are known for
$\Gamma=\Z_p^r$ \citep{LetzterMorrison,ChristophKMS}.

The conjecture $R(d)=d$ would recover the sharp zero-sum bound above: the
lifting construction would give
\[
  n(\Gamma)\le|\Gamma|+1
\]
for every finite group $\Gamma$.  This is best possible as a function of the
group order, since the standard cyclic-group construction gives
$n(\Z_q)\ge q+1$.  The converse does not follow: group labelings produce only
right translations $x\mapsto xa$, whereas the rainbow cycle problem allows
arbitrary maps between classes.  The conjecture $R(d)=d$ remains open.

\section*{Acknowledgements}

The author would like to thank Noga Alon and Sriram Gopalakrishnan for helpful
discussions. This work was conducted while the author was a research intern at
JP Morgan AI Research.

\noindent\textbf{Declaration of AI Use:} ChatGPT 5.6 Sol Ultra was used to
assist with developing the proof, searching the literature, and drafting and
editing the manuscript. The author takes full responsibility for all
mathematical arguments, references, and claims. The author has carefully
reviewed the manuscript for clarity, organization, and accurate attribution
of prior work.

\end{document}